\documentclass[conference]{IEEEtran}
\IEEEoverridecommandlockouts

\usepackage{cite}
\usepackage{amsmath,amssymb,amsfonts}
\usepackage{algorithmic}
\usepackage{graphicx}
\usepackage{textcomp}
\usepackage{xcolor}

\usepackage{amsthm}
\usepackage{verbatim}
\usepackage{tikz}
\usetikzlibrary{arrows,shapes.misc,chains,scopes}
\usetikzlibrary{calc}
\usepackage{pgfplots}
\usepackage[dvipsnames]{xcolor}
\usepackage{acronym}
\usepackage[acronym]{glossaries} \makeglossaries 

\usepackage{resources/fst}

\def\BibTeX{{\rm B\kern-.05em{\sc i\kern-.025em b}\kern-.08em
    T\kern-.1667em\lower.7ex\hbox{E}\kern-.125emX}}

\newacronym{awgn}{AWGN}{additive white Gaussian noise}
\newacronym{ask}{ASK}{amplitude-shift keying}
\newacronym{brgc}{BRGC}{binary reflected Gray code}
\newacronym{ccdm}{CCDM}{constant composition distribution matcher}
\newacronym{dm}{DM}{distribution matcher} 
\newacronym{ess}{ESS}{enumerative sphere shaping}
\newacronym{fec}{FEC}{forward error correction}
\newacronym{fer}{FER}{frame error rate}
\newacronym{ftf}{f2f}{fixed-to-fixed}
\newacronym{isi}{ISI}{inter-symbol interference}
\newacronym{nbc}{NBC}{natural binary code}
\newacronym{ofdm}{OFDM}{orthogonal frequency-division multiplexing}
\newacronym{pas}{PAS}{probabilistic amplitude shaping}
\newacronym{pdm}{PDM}{product distribution matcher}
\newacronym{pmf}{pmf}{probability mass function}
\newacronym{rv}{RV}{random variable}

\begin{document}
\title{Divergence-Minimizing Distribution Matching for Parallel Channels and Channels with Memory
}

\author{\IEEEauthorblockN{Francesca Diedolo and Gerhard Kramer}
\IEEEauthorblockA{\emph{Institute for Communications Engineering} \\
\emph{School of Computation, Information and Technology}\\
\emph{Technical University of Munich, Munich, Germany} \\
francesca.diedolo@tum.de, gerhard.kramer@tum.de}
}

\maketitle

\begin{abstract}
Theory for distribution matchers (DMs) is extended to distributions with memory. The divergence-minimizing DM is shown to select the sequences with the highest target probability, as in the memoryless case. A scaling law for divergence is extended to distributions driven by innovation processes. The theory is applied to parallel additive white Gaussian noise channels. A modified enumerative sphere-shaping (ESS) method with a weighted energy constraint is used in implementations. An illustrative example with three channels shows that joint ESS across channels reduces the rate loss by a large factor compared to product DMs at short blocklengths. The gains are confirmed by simulations with probabilistic amplitude shaping and a 5G-NR low-density parity-check code.
\end{abstract}


\section{Introduction}
\IEEEPARstart{A}{} \gls{dm} is an invertible mapping from source symbols to channel symbols, with the goal of making the channel symbols have a desired distribution \cite{bocherer2015bandwidth}. We consider fixed-to-fixed length \glspl{dm}. 

\Glspl{dm} can be classified as direct or indirect \cite{calderbank2002nonequiprobable}.
``Direct'' refers to starting with a target distribution and designing an algorithm to approach it. A prominent example is a \gls{ccdm} \cite{schulte2015constant}, where the \gls{dm} output sequences have the same empirical distribution, or type. Other examples are multiset-partition distribution matchers \cite{fehenberger2018multiset} and \glspl{pdm} \cite{bocherer2017high,steiner2018approaching}.
``Indirect'' refers to starting with a target rate and requiring the codewords to lie in a sphere with a specified energy. This is referred to as sphere-shaping. Examples of indirect algorithms are shell mapping \cite{laroia1994optimal} and \gls{ess} \cite{willems1993pragmatic}.

Most \gls{dm}s are designed for independent and identically distributed (iid) symbols, which is reasonable for memoryless channels.
However, for channels with memory, such as \gls{isi} channels, input distributions should have memory. Similarly, for parallel channels with different noise levels, the target distributions can be independent, but are not identically distributed across channels.

This paper extends results for \gls{dm}s from iid to general distributions, and designs schemes for parallel \gls{awgn} channels. The paper is organized as follows. Sec.~\ref{sec:preliminaries} introduces notation and distribution matching. Sec.~\ref{sec:properties} extends iid theory to distributions with memory. Sec.~\ref{sec:sphere-shaping} reviews sphere shaping and the \gls{ess} algorithm. Sec.~\ref{sec:parallel-channels} applies the theory to parallel \gls{awgn} channels, and develops a joint \gls{ess} algorithm for such channels. Sec.~\ref{sec:simulations} provides simulation results. Sec.~\ref{sec:conclusions} concludes the paper.

\section{Preliminaries}
\label{sec:preliminaries}
\subsection{Notation and Information Measures}
Sets are written with calligraphic letters $\cA$ and their cardinality as $|\cA|$. 
$1(.)$ is the indicator function that takes on the value 1 if its argument is true, and is 0 otherwise.

\Glspl{rv} are written with uppercase letters such as $A$, and their realizations with the respective lowercase letters. \Glspl{pmf} of $A$ are denoted by $P_A$ or $Q_A$, where $P_A$ and $Q_A$ usually refer to the actual and target distribution, respectively. We often omit subscripts on distributions when their arguments specify the \glspl{rv}. A sequence of \glspl{rv} is written as $A^n = A_1 ,\dots,A_n$.

The self-information of $a$ with respect to $P$ is
\begin{equation}
    i_P (a) = -\log_2P(a).
\end{equation}
The expectation of $A$ with respect to $P$ is 
\begin{equation}
    \bbE_P[A] = \sum\nolimits_{a\in \text{supp}(P)} P(a) \, a
\end{equation}
where $\text{supp}(P)$ is the support of $P$, i.e. the set of $a\in\cA$ with nonzero probability.  The divergence and cross entropy of $P$ and $Q$ are the respective
\begin{equation}
    \bbD(P||Q) = \bbE_P\left[\log_2 \frac{P}{Q}\right], \quad
    \bbX (P|| Q) = \bbE_P\left[\log_2 \frac{1}{Q}\right].
\end{equation}

\subsection{Distibution Matching}
A \gls{dm} is an invertible mapping from $k$ bits to $n$ symbols that approximates a target distribution $Q_{A^n}$. We write
\begin{equation}
    a^n = f_{\text{DM}}(b^k) \in \cA^n
\end{equation}
where $\cA$ is the output alphabet. The mapping efficiency is measured by the rate $R = k/n$. The output distribution is written as $P_{{A}^n}$, and the matching quality is measured by the divergence $\bbD(P_{A^n}||Q_{A^n})$. We usually assume that the bits $B^k$ are uniformly distributed over $\{0,1\}^k$.

The set $\cS$ of $a^n$ generated by a \gls{dm} is called the codebook, and its elements are called codewords. For uniformly distributed $B^k$, each codeword occurs with probability $2^{-k}$.

\section{Optimal Distribution Matchers}
\label{sec:properties}
\subsection{Minimizing Divergence}
Optimal \gls{dm}s for iid $Q_{A^n}=(Q_A)^n$ were studied in \cite{Boecherer-IT16,schulte2017divergence,kramer2021divergence,schulte2022invertible}. The following result extends \cite[Proposition 3]{schulte2022invertible}
to target distributions with memory.

\begin{theorem}\label{thm:divmin}
The DM code $\cS$ that minimizes $\bbD (P_{A^n} || Q_{A^n})$ has all sequences $a^n$ which satisfy 
\begin{equation}
    Q_{A^n}(a^n) \geq 2^{-nI}
\end{equation}
for some $I$. Equivalently, $\cS$ has all sequences satisfying
\begin{equation}
   \frac{1}{n} i_{Q_{A^n}} (a^n) \leq I.
\end{equation}
\end{theorem}

\begin{proof}
The proof follows the steps of the iid case, except that it uses self-information rather than cross-entropy.

Consider $\cS= \cS' \cup \cS''$, where $\cS'$ contains all sequences with $\frac{1}{n} i_{Q_{A^n}} (a^n) \leq \hat{I}$ and $\cS''$ contains exactly $\ell$ sequences with $\frac{1}{n} i_{Q_{A^n}} (a^n) = I$, with $\hat{I}< I$. We have
\begin{align}
    &\bbD (P_{A^n}|| Q_{A^n}) =  \sum_{a^n \in \mathcal{S}} \frac{1}{|\mathcal{S}|}\log_2 \frac{1/|\mathcal{S}|}{Q(a^n)} \nonumber \\
    &= - \log_2 (|\cS'|+ \ell) + \frac{1}{|\cS'|+\ell} \sum_{a^n \in \cS'} -  \log_2 Q(a^n)  \nonumber \\ 
    &\qquad \qquad \qquad + \frac{1}{|\cS'|+\ell} \sum_{a^n \in \cS''} - \log_2 Q(a^n) \nonumber \\
    &  = - \log_2 (|\cS'|+ \ell) +  \frac{n |\cS'| \bar{I}}{|\cS'|+\ell}  + \frac{n \ell I }{|\cS'|+\ell} 
\end{align}
with $\bar{I}= \frac{1}{n |\cS'|}  \sum_{a^n \in \cS'} i_{Q_{A^n}}(a^n)$. Taking derivatives with respect to $\ell$ (considering $\ell$ as continuous) yields:
\begin{align}
    \frac{\partial }{\partial \ell}\bbD (P_{A^n}|| Q_{A^n}) &= \frac{-1}{\ln 2 (|\cS'|+ \ell)} + \frac{n|\cS'|}{(|\cS'|+ \ell)^2} \Delta I \\
    \frac{\partial^2 }{\partial \ell^2}\bbD (P_{A^n}|| Q_{A^n}) &= \frac{1}{\ln 2 (|\cS'|+ \ell)^2} - \frac{2 n|\cS'|}{(|\cS'|+ \ell)^3} \Delta I
\end{align}
where $\Delta I = I - \bar{I}>0$. The first derivative is zero only at
\begin{equation}
    \ell_0 = |\cS'|(n \Delta I(\ln 2)  - 1)
\end{equation}
and we have $\ell_0>-|\cS'|$. The second derivative is negative at $\ell_0$, so $\ell_0$ is a maximizer. Let $\hat{\ell} \in \{0, \dots , \ell_{max}\}$ be the integer that minimizes $\bbD (P_{A^n}|| Q_{A^n})$, where $\ell_{max}$ is the number of $a^n$ with self-information $I$. There are three possible cases:
\begin{align}
    \ell_0 \in [0, \ell_{max}], \quad
    \ell_0 <0, \quad
    \ell_0 > \ell_{max}.
\end{align}
In all cases, we have $\hat{\ell} = 0$ or $\hat{\ell} = \ell_{max}$. 
\end{proof}

An optimal \gls{dm} code $\cS$ thus has all sequences with high probability. Moreover, for sequences with the same probability, they are either all in the set or none of them are in the set. For binary iid distributions, this corresponds to type sets.

\subsection{Divergence Scaling}
The divergence of optimal \gls{dm}s for memoryless sources grows as $\frac{1}{2}\log_2 n$ with $n$ \cite{schulte2017divergence,kramer2021divergence,schulte2022invertible}. We wish to extend this result to sources with memory. Consider innovation-driven processes, i.e., processes generated by mappings
\begin{equation}
    A^n = f(Z^n)
\end{equation}
where $Z^n$ is a sequence of iid samples. Such processes can be represented using a framework called Bernoulli schemes \cite[Chapter 7]{doukhan2018stochastic}, which includes autoregressive and Markov processes. For example, it includes the distribution of the process $A_i = A_{i-1} \oplus Z_i$, where $Z_i$ is iid Bernoulli and $A_0$ is fixed and known. We have the following result.

\begin{theorem} \label{thm:scaling}
Consider a target probability $Q_{A^n}$ with $A^n = f(Z^n)$ where $f(\cdot)$ is invertible and $Z^n$ is drawn from an iid distribution. The divergence $D(P_{A^n}||Q_{A^n})$ of the optimal DM code $\cS$ grows as 
$\frac{1}{2}\log_2 n$ with n.
\end{theorem}
\begin{proof}
 The proof follows from the data processing inequality:
 \begin{align}
    \bbD(P_{A^n}|| Q_{A^n}) \le \bbD(P_{Z^n}|| Q_{Z^n})
\end{align}
with equality if $f(\cdot)$ is invertible.
\end{proof}

Note that the codewords of the optimal codebook for $Q_{A^n}$, which has the most likely $a^n$ wrt $Q_{A^n}$, can be computed by applying $f(.)$ to the most likely $z^n$ wrt $Q_{Z^n}$.

\section{Sphere Shaping}
\label{sec:sphere-shaping}
Sphere shaping realizes a desired target distribution by using sequences that lie on or within an energy sphere. More precisely, sphere shaping accepts only those sequences with energy less than or equal to a specified $E_\text{max}$. 
The induced distribution converges to a discrete Gaussian with the same mean and variance as the sampled distribution.

Consider the alphabet $\cA$. The set of valid sequences is
\begin{equation}
    \cA ^\bullet = \left\{ a_1, \dots , a_n \Big| \sum\nolimits_{i=1}^n a_n ^2 \le E_\text{max}
    \right\} \subseteq \cA ^n .
\end{equation}
For $k$ bits, the maximum energy $E_{\text{max}}$ is chosen as the smallest value such that $|\cA^\bullet| \geq 2^k$.
The codebook $\cS \subseteq \cA ^\bullet $ has size $|\cS| = 2^k$.

We later focus on the \gls{ask} alphabet $\cA=\{ 1, 3, \dots, \sqrt{E_\text{max}}\}$, for which the possible energy levels are $\{n, n+8, \dots, E_{\text{max}} \}$. The maximum energy is
\begin{equation}
    E_\text{max} = n + 8 (L_s-1),
\end{equation}
where $L_s$ is the number of energy shells.
From a geometric perspective, all sequences in $\cA^\bullet$ lie inside or on the surface of a sphere with radius $\sqrt{E_\text{max}}$.
\gls{ccdm} instead has all codewords with the same energy and lying on the surface of a sphere; see Fig.~\ref{fig:spheres}. For large $n$, sphere shaping and \gls{ccdm} become the same in terms of sequence selection and synthesized probability. The reason is that the sphere volume concentrates near the surface for large $n$, a phenomenon known as sphere-hardening.

\begin{figure}
    \centering
    \vspace*{3pt}
    \includegraphics{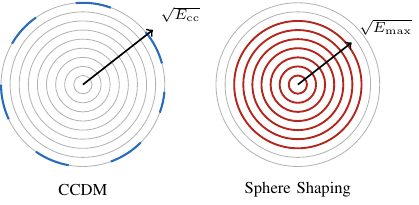}~ 
    \caption{\gls{ccdm} vs. sphere shaping.}
    \label{fig:spheres}
\end{figure}

\subsection{Enumerative Sphere Shaping} \label{subsec:ess}
\gls{ess} indexes the sequences inside a sphere with a specified radius $\sqrt{E_\text{max}}$ \cite{willems1993pragmatic}. The method requires that the $a^n$ in $\mathcal{A}^{\bullet}$ admit a lexicographical ordering, and associates each $a^n$ with an index $i(a^n)$ that specifies the number of sequences that precede it in the ordering. The algorithm constructs a trellis with $n+1$ stages, where each branch corresponds to an amplitude in $\mathcal{A}$. Each trellis path represents a sequence $a^n$ in $\mathcal{A}^{\bullet}$.

Consider a node at stage $s$ with energy $e$, and
let $T_s^e$ be the number of paths that originate from this node. For example, $T_0^0$ is the number of paths starting from the first node (stage $0$, energy $0$), and it counts all the sequences in the trellis. By initializing the last stages with $T_n^e =1$, one may recursively calculate
\begin{equation}
    T_s^e = \sum\nolimits_{a\in\cA} T_{s+1}^{e+a^2}.
    \label{eq:trellis_rec}
\end{equation}
The index of a sequence can be obtained by counting the number of paths that are directed to  (lexicographically) smaller symbols via Cover's formula \cite{cover1973enumerative}: 
\begin{equation}
    i(a^n) = \sum_{s=1}^n \: \sum_{b\in \cA, \, b<a_s} T_s^{b^2 + \sum_{j=1}^{s-1}a_j^2} .
    \label{eq:cover}
\end{equation}

\textbf{Shaping:} The shaping input is a sequence of $k$ bits that is interpreted as an index $i$, $ 0 \leq i < T_0^0$. The first index is initialized as $i_1 = i$. Next, for $s=1, \dots, n$, we select $a_s$ s.t. 
\begin{equation}
   \sum\nolimits_{b<a_s} T_s^{b^2 + \sum_{j=1}^{s-1}a_j^2} \ge i_s < \sum\nolimits_{b\leq a_s} T_s^{b^2 + \sum_{j=1}^{s-1}a_j^2}
\end{equation}
and the index is updated:
\begin{equation}
    i_{s+1} = i_s -  \sum\nolimits_{b<a_s} T_s^{b^2 + \sum_{j=1}^{s-1}a_j^2}.
\end{equation}
The output is the sequence $a_1, a_2, \dots, a_n$, the rate is $R=k/n$ bits/symbol, and the codebook is $\cS_{\text{ess}}$ where $\cS_{\text{ess}} \subseteq \cA^\bullet$.

\textbf{Deshaping:}
The deshaping input is the sequence $a_1, a_2, \dots, a_n$. Initialize $i_{n+1} = 0$ and recursively compute
\begin{equation}
    i_s = \sum\nolimits_{b<a_s} T_s^{b^2 + \sum_{j=1}^{s-1}a_j^2} + i_{s+1}.
\end{equation}
for $s=n, n-1, \dots ,1$. The output index is $i=i_1$.

\subsection{Storage and Computational Complexity} \label{subsec:complexity}
ESS shaping and deshaping require storing the entire trellis, giving storage of $\cO(n^3)$ bits \cite{gultekin2019enumerative}. The computational complexity is on the order of $n$ bit operations per output symbol, growing linearly with $n$ \cite{gultekin2019enumerative}. One can reduce the storage to $\cO(n^2)$ bits with bounded precision\cite{gultekin2018approximate}, and the computational complexity to a constant that depends on the number of bits for the mantissa. One may further reduce complexity by using suboptimal methods, such as quantized \gls{ess} \cite{Savov-Runge-ISIT25}.

\section{Distribution Matchers for Parallel Channels}
\label{sec:parallel-channels}
\subsection{Channel Model and Power Allocation}
Consider $L$ parallel \gls{awgn} channels. The channel input-output relation is
\begin{equation}
  Y_\ell = h_\ell X_\ell + Z_\ell, \quad \ell = 1,2,\dots,L
  \label{eq:channel}
\end{equation}
where the coefficients $h_\ell$ are known at the transmitter and receiver, and the $Z_\ell$ are independent zero-mean Gaussian with unit
variance. We use the channels $n$ times, i.e., there are $Ln$ channel uses in total, and impose the block power constraint
\begin{equation}
  \frac{1}{Ln} \sum\nolimits_{\ell=1}^{L} \sum\nolimits_{i=1}^{n}\mathbb{E}[X_{\ell,i}^2] \le P
  \label{eq:sum_power}
\end{equation}
where $X_{\ell,i}$ is the $i$th symbol of channel $\ell$. The average signal-to-noise ratio over the channels is thus $P$.

The capacity $\frac{1}{L}\sum_\ell \frac{1}{2}\log_2(1+h_\ell^2 P_\ell)$ is achieved by Gaussian signaling and waterfilling power allocation with $\lambda_{\mathrm{wf}}$ chosen so that
\begin{equation}
  P_\ell^\star = \left[\frac{1}{\lambda_{\mathrm{wf}}} - \frac{1}{h_\ell^2}\right]^{+},
  \quad
  \frac{1}{L}\sum\nolimits_{\ell=1}^{L} P_\ell^\star = P .
  \label{eq:waterfilling}
\end{equation}
The rate for $P_\ell^\star>0$ is thus $C_\ell = \tfrac{1}{2}\log_2\!\left(h_\ell^2/\lambda_{\mathrm{wf}}\right)$. The constellation sizes are usually chosen as $2^{m_\ell}$, where
\begin{equation}
  m_\ell = \left\lfloor C_\ell \right\rfloor + 1.
\end{equation}
The additional bit keeps the spectral efficiency loss small  \cite{steiner2018approaching}. We use $2^{m_\ell}$-\gls{ask} with
\begin{equation}
  \mathcal{A}_\ell = \left\{1,3,\dots,2^{m_\ell}-1\right\}.
\end{equation}
The channel input is $X_\ell = S_\ell\Delta_\ell A_\ell$, where $\Delta_\ell>0$ is a
scaling, $S_\ell$ is uniform on $\{\pm1\}$, and $S_\ell$ is independent of
$A_\ell$. 

\subsection{Target Distributions}
We use discrete Gaussian distributions \cite{kschischang1993optimal}
\begin{equation}
	Q_{X_\ell}(x) \;=\;
	\frac{\exp\!\left(-\nu_\ell\,x^2\right)}
	{\sum_{x'\in\cX_\ell}\exp\!\left(-\nu_\ell\,x'^2\right)},
	\quad \nu_\ell > 0
	\label{eq:mb-transmit}
\end{equation}
which induce the amplitude distributions
\begin{equation}
  Q_{A_\ell}(a) \propto \exp\!\left(-\nu_\ell\Delta_\ell^2 a^2\right),
  \quad a\in\mathcal{A}_\ell .
  \label{eq:QA}
\end{equation}
The inputs are independent:
\begin{equation}
	Q_{A^{Ln}}\!\left(a^{Ln}\right)
	= \prod\nolimits_{\ell=1}^{L} \prod\nolimits_{i=1}^{n} Q_{A_\ell}\!\left(a_{\ell,i}\right)
	\label{eq:target}
\end{equation}
\enlargethispage{-4pt}%
and Theorem~\ref{thm:divmin} selects all $a^{Ln}$ with
\begin{equation}
  \sum\nolimits_{\ell=1}^{L} w_\ell\,
  \sum\nolimits_{i=1}^{n} a_{\ell,i}^2 
  \le E_{\max}
  \label{eq:weighted}
\end{equation}
where the energy $E_{\max}$ is chosen to give a desired rate, and the weight $w_\ell:=\nu_\ell\Delta_\ell^2$ modifies the powers $a_{\ell,i}^2$. 

It remains to select the $w_\ell$, for which we can choose $\Delta_\ell$ and $ \nu_\ell$ for convenient implementation. We choose $\Delta_\ell=\Delta/h_\ell$ to give the received constellations a common spacing relative to the noise across all channels \cite{bocherer2017high, steiner2018approaching}. The transmit power is then proportional to $\sum_\ell \mathbb{E}[A_\ell^2]/h_\ell^2$, which one minimizes at a target rate $R_{\mathrm{dm}}$ \cite[eqs.~(43)--(44)]{bocherer2017high}:
\begin{equation}
  \min_{P_{A_1},\dots,P_{A_L}} \sum\nolimits_{\ell=1}^{L}\frac{\mathbb{E}[A_\ell^2]}{h_\ell^2}
  \quad \text{s.t.} \quad
  \frac{1}{L}\sum\nolimits_{\ell=1}^{L}\mathbb{H}(A_\ell) = R_{\mathrm{dm}} .
  \label{eq:minpower}
\end{equation}
The Lagrangian derivative separates over $\ell$ and $a\in\mathcal{A}_\ell$, and one obtains
\begin{equation}
  P_{A_\ell}(a) \propto \exp\!\left(-\mu\,\frac{a^2}{h_\ell^2}\right)
  \label{eq:MBsol}
\end{equation}
implying that $w_\ell \propto 1/h_\ell^2$ with $\nu_\ell=\nu$ for all $\ell$.

We remark that \eqref{eq:MBsol} seems different than the waterfilling solution \eqref{eq:waterfilling}. However, the two approaches are similar for the high-power scenario studied below.

\subsection{Joint Enumerative Sphere Shaping}
\label{subsec:jess}
Consider the target distribution \eqref{eq:target} where $Q_{A_\ell}$ is given by \eqref{eq:MBsol}. We apply joint ESS (J-ESS) across the channels and assign to each channel a rounded weight
\begin{equation}
    w_\ell~=~\left\lfloor \frac{1/h_\ell^2}{\min_j 1/h_j^2} \right\rceil
\end{equation}
so the trellis has integer branch costs.
We build a trellis with $L$ stages for each time index $i$, for a total of $L n+1$ stages. Stage $s$ emits one amplitude for channel $\ell(s)$, the channels being visited round-robin, and a branch carrying amplitude $a\in\mathcal{A}_{\ell(s)}$ costs $w_{\ell(s)}\,a^2$. The node counts satisfy the same backward recursion as \eqref{eq:trellis_rec}:
\begin{equation}
  T_s^{\,e} = \sum\nolimits_{a\in\mathcal{A}_{\ell(s)}} T_{s+1}^{\,e + w_{\ell(s)} a^2}
  \label{eq:trellis_rec_joint}
\end{equation}
with $T_{Ln}^{\,e} = 1(e \le E_{\max})$; indexing, shaping and deshaping remain the same.

Fig.~\ref{fig:joint_ess_trellis} shows an example joint trellis for $L=n=2$. Channel~1 has alphabet $\mathcal{A}_1=\{1,3,5\}$ and weight $w_1=1$, channel~2 has $\mathcal{A}_2=\{1,3\}$ and weight $w_2=4$, and $E_{\max}=30$. Branch costs are $\{1,9,25\}$ on channel~1 and $\{4,36\}$ on channel~2, and the trellis admits four valid amplitude sequences. The storage and complexity are given in Sec.~\ref{subsec:complexity} with $n$ replaced by $Ln$.

\begin{figure}[t]
    \centering
    \includegraphics{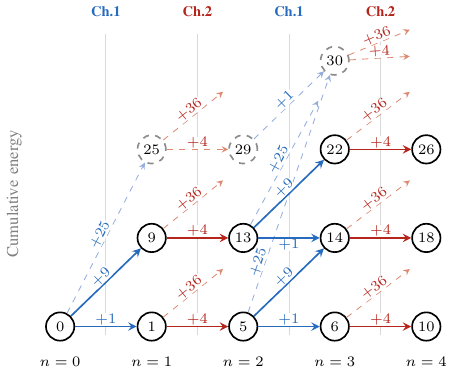}~ 
    \caption{J-ESS trellis for the considered example. The number inside the nodes represents the energy level.}
    \label{fig:joint_ess_trellis}
\end{figure}

\section{Simulation Results}
\label{sec:simulations}
We illustrate the performance for $L = 3$ channels with coefficients $h_\ell = 2,\,1,\,0.5$ and $w_\ell = 1,\,4,\,16$. Let $m_\ell = 5,\,4,\,3$ so the channels carry $32$-, $16$- and $8$-\gls{ask}. 

\subsection{Rate Loss}
We measure the rate loss in bits per symbol as
\begin{equation}
  R_{\mathrm{loss}}
  = \frac{1}{L}\sum\nolimits_{\ell=1}^{L}\mathbb{H}\!\left(P_{A_\ell}\right)
    - \frac{k}{Ln}.
  \label{eq:rateloss_parallel}
\end{equation}
Fig.~\ref{fig:rateloss} compares $R_{\mathrm{loss}}$ for \glspl{pdm}, per-channel \gls{ess} with $L$ independent sphere shapers, and J-\gls{ess}. All shapers have $k$ bits of input. As expected, J-\gls{ess} exhibits the smallest rate loss, and
the gains are larger at short blocklengths. At $n = 10$, the shapers lose $0.400$, $0.118$, and $0.063$ bits per use, respectively; at $n=50$, they lose $0.107$, $0.044$ and $0.019$ bits per use. J-\gls{ess} thus reduces the rate loss by a factor of 5--6 compared to PDM, and a factor of 2--3 compared to per-channel ESS.

\begin{figure}[t]
    \centering
    \includegraphics{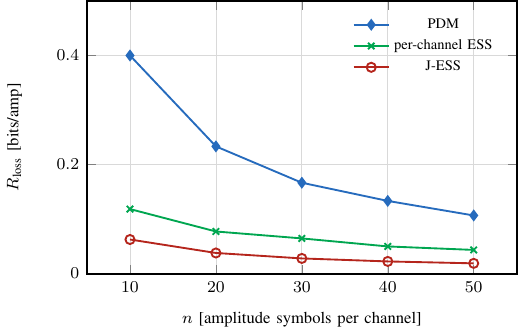}
     \vspace{-0.5cm}
    \caption{Rate loss for parallel AWGN channels.}
    \label{fig:rateloss}
\end{figure}

The gains arise for two reasons. First, joint shaping indexes a single codebook over all $Ln$ amplitudes instead of $L$ codebooks of length $n$, which amortizes the finite-length penalty in the same way PDM's shared component DMs do. Second, unlike PDM, J-\gls{ess} is not restricted to a product distribution over bit levels: by Theorem~\ref{thm:divmin}, it is the divergence-minimizing code for the target \eqref{eq:target} at that rate.

\subsection{Simulations with 5G-NR LDPC Codes}
We use \gls{pas} \cite{bocherer2015bandwidth} and the 5G-NR low-density parity-check (LDPC) code with base graph~1. The systematic block has $\gamma L n$ sign-data bits followed by $\sum_\ell (m_\ell - 1)n$ amplitude bits, where $\gamma\in[0,1]$. The signs are ordered so that 5G-NR's mandatory puncturing of the leading $2Z$ systematic bits removes uniform sign bits rather than shaped amplitude bits. 
With rate 5/6, this gives $\gamma = 1/3$ and overall blocklength $n_{\mathrm{code}} = 12n$. 
Amplitudes are generated with the \gls{nbc} labeling and are relabelled with the \gls{brgc} before encoding. We use belief-propagation decoding with up to $100$ iterations, and every \gls{fer} point in Fig.~\ref{fig:fer} is simulated with at least $100$ frame errors.

\begin{figure} [t]
    \centering
    \includegraphics{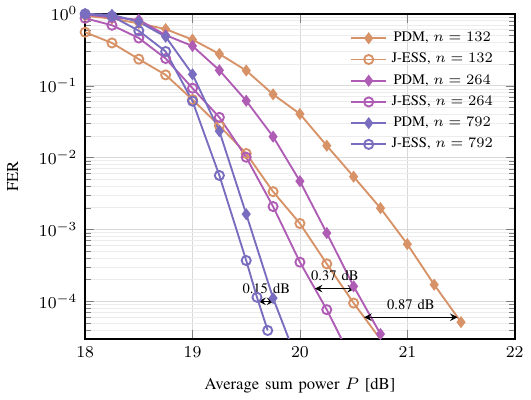}
    \vspace{-0.5cm}
    \caption{\Gls{fer} comparisons of PDM and J-ESS.}
    \label{fig:fer}
\end{figure}

We study three blocklengths: $n = 44,\,88,\,264$ (or $Ln = 132,\,264,\,792$) for each channel with lifting sizes $Z = 20,\,40,\,120$. For each $n$, the number $k$ of information bits is set to the largest input length that \gls{pdm} can support.
\gls{pdm} uses one binary CCDM per amplitude bit level. 
CCDM$_{i}$ has binary output length $n_i = n |\{\ell: m_\ell \geq m-1\}|$, with $i=1 \dots m-1$ and target Bernoulli parameter $p_i$. This results in sequences with a composition of  $c_i \approx  n_i p_i$ ones, so
\begin{equation}
     k = \sum\nolimits_{i=1}^{m-1}\left\lfloor \log_2 \binom{n_i}{c_i} \right\rfloor 
\end{equation}
giving $k = 334,687,2091$ and the average rates
$R_\text{avg}=2.864,2.936,2.974$~bpcu over all $Ln$ uses (so multiply by three to obtain the sum rate across the three channels).

J-\gls{ess} has the same $k$ at the same $n$, so the two schemes are rate-matched for each blocklength. Across blocklengths, they are not matched: the \gls{ccdm} rate loss is nearly constant in bits while $Ln$ shrinks sixfold, so $R_\text{avg}$ drifts by $0.11$~bpcu and curves at different $n$ should not be compared.

J-\gls{ess} shapes in blocks of $3 \times 44 = 132$ amplitudes for every $n$, so the shaping blocklength is fixed while the coding blocklength varies. Since $k$ does not divide evenly across blocks, each configuration uses two trellises with input lengths $k$ and $k+1$. Fig.~\ref{fig:fer} shows the \glspl{fer}.

\subsection{Complexity}
J-\gls{ess} has higher complexity than the other schemes because it jointly shapes all channels using a single trellis over the $N = 3\times 44 = 132$ amplitudes of a block. 
In our setup, this trellis has about $2.7\times 10^{6}$ nodes, each storing
an integer of up to $k \approx 350$ bits. 
In general, the storage grows as $\mathcal{O}(N^3)$ bits \cite{gultekin2019enumerative}, or $\mathcal{O}(N^2)$ bits with
bounded precision \cite{gultekin2018approximate}, and shaping proceeds sequentially over
the $N$ amplitudes. 
PDM instead runs $m-1 = 4$ binary CCDMs in parallel, with small storage and a linear complexity in the output length
\cite{bocherer2017high, steiner2018approaching}. 
Nevertheless, J-\gls{ess} may be attractive for short blocks where the rate loss dominates; see Fig.~\ref{fig:fer}.

\section{Conclusions}
\label{sec:conclusions}
We extended theory for \glspl{dm} to distributions with memory. Theorem~\ref{thm:divmin} characterizes the divergence-minimizing DM code for any $Q_{A^n}$ by collecting the sequences with the largest target probability. Theorem~\ref{thm:scaling} extends the $\frac{1}{2}\log_2 n$ divergence scaling to innovation-driven processes.

Theorem~\ref{thm:divmin} gives a weighted energy constraint for parallel \gls{awgn} channels. We used this insight to extend \gls{ess} by using multiple stages at each time index and making the branch-metric stage-dependent. The resulting shaper can be used with PAS. J-\gls{ess} reduces the rate loss compared to \gls{pdm} by a large factor for an illustrative channel and short blocklengths. The gains carry over to the coded performance with PAS and a 5G-NR LDPC code. These gains come at the cost of a higher trellis complexity than PDM. Future work may reduce complexity, e.g., by using quantized \gls{ess} \cite{Savov-Runge-ISIT25}.

\section*{Acknowledgments}
This work was supported by the Deutsche Forschungsgemeinschaft (DFG) through research grant {KR 3517/13-1}.

\bibliographystyle{IEEEtran}
\bibliography{references}

\vfill

\end{document}